\documentclass[reqno]{amsart}
\usepackage{bbm}
\usepackage{amsfonts}
\usepackage{mathrsfs}
\usepackage{hyperref}
\usepackage{amssymb}
\usepackage{multirow}
\usepackage{CJK}
 \newtheorem{thm}{Theorem}[section]
 \newtheorem{cor}[thm]{Corollary}
 \newtheorem{prop}[thm]{Corollary}
 
 \newtheorem{conj}[thm]{Conjecture}
 \newtheorem{lem}[thm]{Lemma}

 \newtheorem{defn}[thm]{Definition}
 \theoremstyle{remark}
 
 \numberwithin{equation}{section}
\DeclareMathOperator{\sign}{sign}

\newcommand{\order}{\textup{order}}
\begin{document}
\title[The simplified weighted sum function and its average sensitivity]
  {The simplified weighted sum function and its average sensitivity}

\author{Jiyou Li}
\address{Department of Mathematics, Shanghai Jiao Tong University, Shanghai, P.R. China}
\email{lijiyou@sjtu.edu.cn}

\author{Chu Luo}
\address{Department of Electronics and Computer Science, University of
Southampton, Southampton, UK} \email{cl7e13@soton.ac.uk}

%\author{Daqing Wan}
%\address{Department of Mathematics, University of California, Irvine, CA 92697-3875, USA}
%\email{dwan@math.uci.edu}

 %\author[ J. Y. Li, D. Q. Wan]
 %{Jiyou Li, Daqing Wan}

%\date{2006 10 10}
%\thanks{Submitted September 8, 2005. Published January 19, 2006.}
\thanks{This work is supported by the National Science Foundation of China
(11001170) and the National Science Foundation of Shanghai Municipal
(13ZR1422500).}
%This work is supported by NSFC, grants
%National Science Foundation of China
%  and \hfill\break\indent
%Y2005A07 from Natural Science Foundation, Shandong Province,
%China} \subjclass[2000]{34B16} \keywords{Singular sublinear
%boundary-value problem; positive solution; \hfill\break\indent
%fixed point theorem; cone; higher order differential equation}

\begin{abstract}

In this paper we simplify the definition of the weighted sum Boolean function which used to be inconvenient to compute and use.
We show that the new function has essentially the same properties as the previous one.
In particular, the bound on the average sensitivity of the weighted sum Boolean function remains unchanged after the simplification.

%As an illustration of applications, we give a combinatorial proof on
%a problem raised by Stanley.

%In particular,  we obtain some positive results in the decoding
%problems arising from Reed-Solomon codes this together with the by
%by using the Weil bound.
%We also try to apply our method to the calculations of the number of
%permutation polynomials over finite fields.
\end{abstract}

\maketitle \numberwithin{equation}{section}
\newtheorem{theorem}{Theorem}[section]
\newtheorem{lemma}[theorem]{Lemma}
\newtheorem{example}[theorem]{Example}
\allowdisplaybreaks

\section{Introduction}
In previous study, the weighted sum function has a complicated structure.
With a residue ring modulo a prime, the explicit definition of this function
can be given using the weighted sum as follows \cite{SZ}. Let $m \in \mathbb{Z}^+ = \{1,2,3,\dots \}$
and prime number $p \geq m$ where no other prime numbers are between $p$ and $m$.
For vector $X=(x_1,x_2,\dots,x_m)\in \mathbb{Z}_2^m$, where $\mathbb{Z}_2= \{0, 1\}$, let $u(X)$ be the
least positive integer which satisfies
$$u(X)= \sum_{k=1}^m kx_k({\rm mod}\  p), 1\leq u(X) \leq p.$$
Then the weighted sum function $g(X)$ is defined as
$$g(X)=\left\{
\begin{array}{ll}
  x_{u(X)}, \ \  1\leq u(X)\leq m;\\
  x_1, \ \  \hbox{otherwise.}\\
    \end{array}
    \right.
    $$

This function was used to study read-once
branching programs by P.
Savick\'{y} and S. \v{Z}\'{a}k \cite{SZ}. It was also used to demonstrate
the exponential improvement from conventional read-once branching programs to quantum ones
by M. Sauerhoff in \cite{S}, see also \cite{SD}.

To simplify the definition of the previous weighted sum function, we define a new function $f(X)$ as follows.
For $X=(x_0,x_1,\dots,x_{m-1})\in \mathbb{Z}_2^m$, denote
$$s(X)= \sum_{k=0}^{m-1} kx_k({\rm mod}\  m),$$
and define the new weighted sum function
$$f(X)=x_{s(X)}.$$

It is worth noting that this new function $f(X)$ is more convenient
to compute and use than $g(X)$. One particular reason for the prime
modulus in the previous function $g(X)$ is that there are nice
results and structures in prime fields. In this paper we call such
$f(X)$  the simplified weighted sum function. Note that when $m$ is
prime then the two definitions are the same.

In this paper it is shown that the simplified function $f(X)$ has many
similar properties as the previous one $g(X)$. For instance, in
\cite{SZ} the authors used $g(X)$ to establish the lower
bound of read-once branching programs. One of the key ingredients in
their proof is that
\begin{thm} [Dias da Silva and Hamidoune, \cite{DH}]
 Let $\epsilon>0$ be fixed. Then, for every large enough $p$ and $A\subseteq \mathbb{Z}_p$
 with $|A|>(2+\epsilon)\sqrt{p}$, and for every $b\in \mathbb{Z}_p$, there is a subset
$B\subseteq A$ such that the sum of the elements of $B$ is equal to $b$.
\end{thm}
We note that the work of Freeze, Gao and Geroldinger \cite{FGG} implies the similar result in $\mathbb{Z}_m$.
\begin{thm} [Freeze, Gao and Geroldinger, \cite{FGG}]
 Let $d$ be the smallest prime divisor of $m$. Then, for every $A\subseteq \mathbb{Z}_m$
 with $|A|>\frac m d+d-2$, and for every $b\in \mathbb{Z}_m$, there is a subset
$B\subseteq A$ such that the sum of the elements of $B$ is equal to $b$.
\end{thm}

%Another example is that in  \cite{SD,S}, the weighted sum functions are showed that whose
%quantum read-once branching programs are exponentially smaller than classical
%randomized ones.   The key lemma in this paper is that

We then determine the average sensitivity of this newly defined function $f(X)$ and
show that it also satisfies the Shparlinski's conjecture \cite{Sh} which says that the average
sensitivity of $f(X)$ is asymptotically $m/2$. We introduce the main concepts of this conjecture in the following.

For an input $X=(x_0,x_1,\dots,x_{m-1})$, the sensitivity $\sigma_{s,
X}(f)$ on $X$ denotes the number of variables such that flipping
one of these variables will shift the value of $f$. Explicitly,
$$\sigma_{s,X}(f)=\sum_{i=0}^{m-1}\left |
 f(X)-f(X^{(i)})\right |,$$
 where $X^{(i)}=(x_0,\dots,x_{i-1},1-x_i,x_{i+1}\dots,x_{m-1})$ is
the vector assignment after flipping the $i$-th coordinate in $X$.
 The sensitivity $\sigma_{s}(f)$ of $f(X)$ denotes the maximum of $\sigma_{s,X}(f)$ on
vector $X$ in $\mathbb{Z}_2^m$ and the average sensitivity
$\sigma_{av}(f)$ is the mean value of sensitivity on every possible
input, i.e.,
$$\sigma_{av}(f)=2^{-m}\sum_{X\in
\mathbb{Z}_2^{m}}\sum_{i=0}^{m-1}\left |
 f(X)-f(X^{(i)})\right |.$$
   Sensitivity, together with a more general concept called block sensitivity, is a useful measure
to predict the complexity of Boolean functions.
%   It also has many applications in diverse fields especially in complexity theory.
It has recently drawn extensive attention, for instance \cite{A, Be, Bd, Bo,
Bu, CG, Li, R, ST, Shi, Sh}. For a good survey on the main unsolved
problems on sensitivity, please refer to \cite{HKP}.

   In \cite{Sh} Shparlinski addressed the average sensitivity problem of the previous weighted sum function $g(X)$
   and obtained a lower bound from a nontrivial bound on its Fourier coefficients using exponential sums methods. He also developed several conjectures on the
   average sensitivity of the weighted sum function
   and the bounds of the Fourier coefficients.
Explicitly, one conjecture was that the average sensitivity of $g(X)$ on $m$ variables is not less than $(\frac 12+o(1))m$. In the same paper he gave a proof that the average sensitivity
 is greater than $\gamma m$, where constant $\gamma$ satisfies $\gamma\approx 0.0575$.

By applying a new sieving technique, in \cite{Li}  the first author gave an asymptotic
counting formulas of the subset sums over prime fields and thus confirmed the
Shparlinski's conjecture on the average sensitivity of the weighted sum function.

In this paper we extend this result for the simplified weighted sum function $f(X)$.
That is, for $f(X)$ with $m$ variables, the average sensitivity of $f(X)$ is exactly
$(1/2+o(1))m$.

In addition, we also compute the weight of $f(X)$.
We prove that the weight of $f(X)$ on $m$ variables  is exactly
$2^{m-1}(1+o(1))$. Thus, $f(X)$ is an asymptotically balanced
function.

This paper is organized as follows. In Section 2 we present a sieve
formula. By applying this formula, we give a series of formulas for
counting subsets sums over cyclic groups in Section 3. The proof of
the main results is given in Section 4. We also list several further questions in Section 5.

\noindent {\bf Notation}. For $x\in\mathbb{R}$, let $(x)_0=1$
 and $(x)_k=x (x-1) \cdots (x-k+1)$
for $k\in $ $\mathbb{Z^+}=\{1,2,3,\ldots\} $. For $k\in
\mathbb{N}=\{0,1,2,\ldots\}$ define the binomial coefficient ${x
\choose k}=\frac {(x)_k}{k!}$.

\section{A distinct coordinate sieving formula}

For the purpose of our proof, we briefly introduce a sieving formula
discovered by Li-Wan \cite{LW2}, which significantly improves the
classical inclusion-exclusion sieving. We cite it here without any
proof. For details and related applications, we refer to
\cite{LW2,LW3}.

Let $S_k$ be the symmetric group on $k$ elements. It is well known
that every permutation $\tau\in S_k$ factorizes uniquely %(up to the
%  order of the factors)
as a product of disjoint cycles and  each
  fixed point is viewed as a trivial cycle of length $1$.
 % For simplicity of the notation, we usually omit the $1$-cycles.
 For $\tau\in S_k$, define $\sign(\tau)=(-1)^{k-l(\tau)}$, where $l(\tau)$ is the number of
 cycles of $\tau$ including the trivial cycles.

% Each element of $X_{\tau}$ is said to be of type $\tau$.

\begin{thm} \label{thm1.0}
Suppose $X$ is a finite set of vectors of length $k$ over an
alphabet set $D$.  Define $\overline{X}=\{(x_1,x_2,\cdots,x_k)\in X
\ | \ x_i\ne x_j, \forall i\ne j\}.$ Let $f(x_1,x_2,\dots,x_k)$ be a
complex valued function defined over $X$ and $F=\sum_{x \in
\overline{X}}f(x_1,x_2,\dots,x_k).\label{1.00}$. Then
\begin{align}
  \label{1.5} F=\sum_{\tau\in S_k}{\sign(\tau)F_{\tau}},
    \end{align}
    where  \begin{align}\label{1.1}
     X_{\tau}=\left\{
(x_1,\dots,x_k)\in X,
 x_{i_1}=\cdots=x_{i_{a_1}},\cdots, x_{l_1}=\cdots=x_{l_{a_s}}
 \right\},
\end{align}
for a permutation $\tau=(i_1i_2\cdots i_{a_1})
  (j_1j_2\cdots j_{a_2})\cdots(l_1l_2\cdots l_{a_s})$
  with $1\leq a_i, 1 \leq i\leq s$ and
    $F_{\tau}=\sum_{x \in
X_{\tau} } f(x_1,x_2,\dots,x_k). $
 \end{thm}

%
%\begin{defn}Let $X\subseteq D^k$ and assume $X$ is symmetric.
%A complex-valued function $f(x_1,x_2,\dots, x_k)$ defined over $X$
%is called normal on $X$ if for any two $S_k$-conjugate elements
%$\tau$ and $\tau'$ in $S_k$ (thus $\tau$ and $\tau'$ have the same
%type), we have
%$$\sum_{x\in X_{\tau}
%} f(x_1,x_2,\dots,x_k)=\sum_{x\in X_{\tau'} }
%f(x_1,x_2,\dots,x_k).$$
% \end{defn}
%
%
%\textbf{Remark.}
%  If $f(x_1,x_2,\dots,x_k)$ is a symmetric function and $X$ is symmetric,
% then $f(x_1,x_2,\dots,x_k)$ must be normal on $X$.
 %and note that $|C_k|=p(k)$, the partition function.
% For $\tau\in S_k$, denote by $\overline{\tau}$
% the conjugacy class determined by $\tau$ and it can
%also be viewed as the set of permutations conjugate to $\tau$.
%Conversely, for given conjugacy class $\overline{\tau}\in C_k$,
%denote by $\tau$ a representative permutation of this class. For
%convenience we usually identify these two symbols. Since two
%permutations in $S_k$ are conjugate if and only if they have the
%same type of cycle structure (up to the order), $C_k$ is exactly the
%set of all partitions of $k$.

Note that the symmetric group $S_k$ acts on $D^k$ naturally by
permuting coordinates. That is, for $\tau\in S_k$ and
$x=(x_1,x_2,\dots,x_k)\in D^k$,  $\tau\circ
x=(x_{\tau(1)},x_{\tau(2)},\dots,x_{\tau(k)}).$
% Before stating a more useful corollary, we need two definitions first.
  A subset $X$ in $D^k$ is said to be symmetric if for any $x\in X$ and
any $\tau\in S_k$, $\tau\circ x \in X $. In particular,  if  $X$ is
symmetric and $f$ is a symmetric function under the action of $S_k$,
we then have the following formula which is simpler than
(\ref{1.5}).
\begin{prop} \label{thm1.1} Let $C_k$ be the set of conjugacy  classes
 of $S_k$.  If $X$ is symmetric and $f$ is symmetric, then
 \begin{align}\label{7} F=\sum_{\tau \in C_k}\sign(\tau) C(\tau)F_{\tau},
  \end{align} where $C(\tau)$ is the number of permutations conjugate to
  $\tau$.
 %  and the summation is through the representative
%permutations in $C_k$.
%%
%If $f$ is strongly normal on $X$,  for given $1 \leq i\leq k$, we
%choose $\tau_i\in S_k$ satisfying $l(\tau_i)=i$, and let
%$$F_{i}=\sum_{x\in X_{\tau_i}
%} f(x_1,x_2,\cdots,x_k)$$ (this is independent of the choice of
%$\tau_i$),  then we have
% \begin{eqnarray}
%F=\sum_{i=1}^{k}(-1)^{k-i}c(k,i)F_{i},
% \end{eqnarray}
%where $c(k,i)$ is the signless Stirling number of the first kind,
%that is,  the number of permutations in $S_k$ with exactly $i$
%cycles.
\end{prop}

For the purpose of evaluating the above summation, we need several
combinatorial formulas. Recall that a permutation $\tau\in S_k$ is
said to be of type $(c_1,c_2,\cdots,c_k)$ if $\tau$ has exactly
$c_i$ cycles of length $i$ and that $\sum_{i=1}^k ic_i=k$. Let
$N(c_1,c_2,\dots,c_k)$ be the number  of permutations in $S_k$ of
type $(c_1,c_2,\dots,c_k)$ and it is well-known  that
$$N(c_1,c_2,\dots,c_k)=\frac {k!} {1^{c_1}c_1! 2^{c_2}c_2!\cdots k^{c_k}c_k!}.$$

\begin{lem} \label{lem2.6}
Define the generating function
\begin{align*}C_k(t_1,t_2,\dots,t_k)= \sum_{\sum
ic_i=k} N(c_1,c_2,\dots,c_k)t_1^{c_1}t_2^{c_2}\cdots t_k^{c_k}.
 \end{align*}
If $t_1=t_2=\cdots=t_k=q$, then we have
\begin{align*}
C_k(q,q,\dots,q) &=\sum_{\sum
ic_i=k} N(c_1,c_2,\dots,c_k)q^{c_1}q^{c_2}\cdots q^{c_k}\nonumber \\
&=(q+k-1)_k.
 \end{align*}
 In another case, if $t_i=q$ for $d\mid i$ and $t_i=s$ for $d\nmid
i$, then we have
\begin{align*}
C_k(\overbrace{s,\cdots,s}^{d-1},q,\overbrace{s,\cdots,s}^{d-1},q,
\cdots) &= \sum_{\sum
ic_i=k} N(c_1,c_2,\cdots,c_k)q^{c_1}q^{c_2}\cdots s^{c_d}q^{c_{d+1}}\cdots\nonumber \\
&= k!\sum_{i=0}^{\lfloor k/d \rfloor}{\frac{q-s}{d}+i-1\choose
\frac{q-s}{d}-1} {s+k-di-1\choose s-1}\\
&\leq k!{s+k+(q-s)/d-1\choose k}\\
&=(s+k+(q-s)/d-1)_k.
 \end{align*}
\end{lem}

\section{Subset Sum Problem in a Subset of the Cyclic Groups}

 Let $\mathbb{Z}_m$ be the cyclic group of $m$ elements.
 Let $D\subseteq \mathbb{Z}_m$ be a nonempty subset of cardinality $n$. Let
$\widehat{\mathbb{Z}}_m$ be the group of additive characters of $\mathbb{Z}_m$, i.e, all the homomorphisms from $\mathbb{Z}_m$  to the nonzero complex numbers $\mathbb{C}^*$.  Note that
$\widehat{\mathbb{Z}}_m$ is isomorphic to $\mathbb{Z}_m$. Define $s_{\chi}(D)=\sum_{a\in
D}\chi{(a)}$ and  $\Phi(D)=\max_{\chi\in \widehat{\mathbb{Z}}_m, \chi\ne
\chi_0}|s_{\chi}(D)|$. Let $N(k, b, D)$ be the number of $k$-subsets
$T\subseteq D$ such that $\sum_{x\in S}x=b$.
%solutions of the equation
%$$x_1+x_2+\cdots+x_k=b, x_i\in D, x_i\neq x_j, i\ne j.$$
In the following theorem we will give an asymptotic bound for $N(k, b,
D)$ which ensures $N(k, b, D)>0$ when $\mathbb{Z}_m-D$ is not too large compared with $\mathbb{Z}_m$.

 \begin{thm}\label{lem1.1}Let $N(k, b, D)$ be defined as above.
\begin{align}\label{4.1}
\left | N(k, b, D)-{m^{-1}} {n \choose k}\right| &\leq  \frac 1 {m}\sum_{1<r\leq m \atop r\mid
m} \phi(r){\frac {n+\Phi(D)}{r}+k-1 \choose k},
 \end{align}
where $d$ is the smallest prime divisor of $m$.
  \end{thm}

\begin{proof}
Let $X=D\times D \times \cdots \times D$ be the Cartesian product of
$k$ copies of $D$.
 Let $  \overline{X} =\left\{ (x_1,x_2,\dots,x_{k} )\in D^k \mid
 x_i\not=x_j,~ \forall i\ne j\} \right\}.$ It is clear that $|X|=n^k$ and
$|\overline{X}|=(n)_k$. Applying the orthogonal relation
$\sum_{\psi\in \widehat{\mathbb{Z}}_m }\psi(a)=0$ for $a\not\equiv 0
 (\bmod \ m)$ and $\sum_{\psi\in \widehat{\mathbb{Z}}_m }\psi(a)=m$ for $a\equiv 0
 (\bmod \ m)$, we have

\begin{align*}
k!N(k, b, D)&={m^{-1}} \sum_{(x_1, x_2,\dots x_k) \in
\overline{X}}
\sum_{\chi\in \widehat{\mathbb{Z}}_m}\chi(x_1+x_2+\cdots +x_k-b)\\
%&={p^{-1}} \sum_{\chi\in \widehat{{\bf F}}_p} \sum_{(x_1, x_2,\dots
%x_k) \in \overline{X}}\chi(x_1+x_2+\cdots +x_k-b)\\
&={m^{-1}} (n)_k+m^{-1} \sum_{\chi\ne \chi_0}\sum_{(x_1,
x_2,\cdots x_k) \in\overline{X}}\chi(x_1)\chi(x_2)\cdots \chi(x_k)\chi^{-1}(b)\\
&={m^{-1}}  {(n)_k}+{m^{-1}} \sum_{\chi\ne
\chi_0}\chi^{-1}(b)\sum_{(x_1,x_2,\dots x_k)
\in\overline{X}}\prod_{i=1}^{k} \chi(x_i).
\end{align*}
Denote $f_{\chi}(x)= f_{\chi}(x_1,x_2,\dots,x_{k})=
\prod_{i=1}^{k}\chi(x_i).$  For
  $\tau\in S_k$,  let
$$F_{\tau}(\chi)=\sum_{x\in X_{\tau}}f_{\chi}(x)=\sum_{x \in X_{\tau}}\prod_{i=1}^{k} \chi(x_i),$$
where $X_{\tau}$ is defined as in (\ref{1.1}). Obviously $X$ is
symmetric and $f_{\chi}(x_1,x_2,\dots,x_{k})$ is normal on $X$.
Applying (\ref{7}) in Corollary \ref{thm1.1}, we get
  \begin{align*}
 k!N(k, b, D)&={m^{-1}} {(n)_k}+{m^{-1}} \sum_{\chi\ne \chi_0}\chi^{-1}(b) \sum_{\tau\in
C_{k}}\sign(\tau)C(\tau) F_{\tau}(\chi),
  \end{align*}
 where $C_{k}$ is the set of conjugacy classes
 of $S_{k}$, $C(\tau)$ is the number of permutations conjugate to $\tau$.  If $\tau$
 is of type $(c_1, c_2, \dots, c_k)$, then
 %$F_{\tau}(\chi)=\sum_{x \in X_{\tau}}\prod_{i=1}^{k}\chi(x_i).$
%For  $\tau\in C_{k}$, assume $\tau$ is of type
%$(c_1,c_2,\dots,c_{k})$, where $c_i$ is the number of $i$-cycles in
%$\tau$ for $1 \leq i\leq k$. Note that $\sum_{i=1}^{k} ic_i=k$.
%Assume
%$$\tau=(i_1)(i_2)\cdots(i_{c_1})(i_{c_1+1}i_{c_1+2})(i_{c_1+3}i_{c_1+4})\cdots
%(i_{c_1+2c_2-1}i_{c_1+2c_2})\cdots,$$
%  and one checks that
%  \begin{align*}
%      X_{\tau}=\left\{
%(x_1,\dots,x_k)\in D^k,
% x_{i_{c_1+1}}=x_{i_{c_1+2}}, \cdots  x_{i_{c_1+2c_2-1}}=x_{i_{c_1+2c_2}},\cdots
% \right\}.
%\end{align*}
  \begin{align*}
F_{\tau}(\chi)&=\sum_{x \in X_{\tau}}\prod_{i=1}^{k} \chi(x_i)\\
&=\sum_{x \in X_{\tau}}\prod_{i=1}^{c_1} \chi(x_i)\prod_{i=1}^{c_2}
\chi^2(x_{c_1+2i})\cdots\prod_{i=1}^{c_k} \chi^k(x_{c_1+c_2+\cdots+k i})\\
 &=\prod_{i=1}^{k}(\sum_{a\in D}\chi^i(a))^{c_i}\\
 &= n^{\sum c_i m_i(\chi)} s_{\chi}(D)^{{\sum c_i (1-m_i(\chi))}},
\end{align*}
where $m_i(\chi)=1$ if $\chi^i=1$ and otherwise $m_i(\chi)=0$.

  Now suppose $\order(\chi)=r$ with $d \leq r\mid m$. Note that
$C(\tau)=N(c_1,c_2,\dots,c_k)$ and by  Lemma \ref{lem2.6} we have
\begin{align*}
&\sum_{\tau\in C_{k}}\sign(\tau)C(\tau) F_{\tau}(\chi)\\
&\leq\sum_{\tau\in C_{k}}C(\tau) n^{\sum c_i m_i(\chi)}
\Phi(D)^{{\sum c_i (1-m_i(\chi))}}\\
&\leq k!{\frac {n+\Phi(D)}{r}+k-1 \choose k}.
\end{align*}

Similarly, if $\order(\chi)$ is greater than $k$,  then
\begin{align*}
\sum_{\tau\in C_{k}}\sign(\tau)C(\tau) F_{\tau}(\chi)\leq k!{
\Phi(D)+k-1 \choose k}.
\end{align*}

 Note that there are $\phi(r)$ characters of order $r$.
 Summing over all nontrivial characters,  we obtain
\begin{align*}
\left | N(k, b, D)-{m^{-1}} {n \choose k}\right| &\leq  \frac 1 {m}\sum_{1<r\leq m \atop r\mid
m} \phi(r){\frac {n+\Phi(D)}{r}+k-1 \choose k},
 \end{align*}
 where $\phi(r)$ is the Euler function.
 This completes the proof.
\end{proof}

\begin{cor}We have
\begin{align*}
\left | N(k, b, D)-{m^{-1}} {n \choose k}\right| &\leq   {\frac {n+\Phi(D)}{d}+k-1 \choose k},
 \end{align*}
 where $d$ is the minimum prime divisor of $m$.
\end{cor}

\begin{cor}If $|D|=m-c$ and $c$ is a positive constant, noting that
$\Phi(D)\leq c$ we have
\begin{align*}
\left | N(k, b, D)-{m^{-1}} {m-c \choose k}\right| &\leq   {\frac
{m}{d}+k-1 \choose k},
 \end{align*}
 where $d$ is the minimum prime divisor of $m$.
\end{cor}

A simple combinatorial arguments on sums of binomial coefficients
gives
  \begin{cor}\label{cor3.4}
 Let $n=m-o(m)$. Let $N(b, D)=\sum_{k=0}^n N(k, b, D)$ be the number of subsets in $D$
 which sums to $b$. Then
 $N(b, D)=\frac {2^n} m (1+o(1)).$
 \end{cor}

\section{Average sensitivity}

In \cite{Li}, the weight and the average sensitivity of
$g(X)$ are computed. We now generalize these results to the
simplified function $f(X)$.
 We first compute the weight of $f(X)$.

\begin{thm}Let $f(X)$ be defined as above. Then we have
$$wt(f)={2^{m-1}}(1+o(1)).$$
In other words, $f(X)$ is an asymptotically balanced function.
\end{thm}

\begin{proof}
  By applying
Corollary \ref{cor3.4} we have
 \begin{align*}
 wt(f(X))=\sum_{X\in Z_2^m}f(X)&=\sum_{s=0}^{m-1}\sum_{X\in Z_2^m, s(X)=s, x_s=1}1 \\
&=\sum_{s=0}^{m-1} N(0,\mathbb{Z}_m\backslash \{s\})\\
&=\sum_{s=0}^{m-1} \frac 1 m  2^{m-1}(1+o(1))\\
 &={2^{m-1}}(1+o(1)) \qedhere
 \end{align*}
\end{proof}

In \cite{Sh} Shparlinski studied $\sigma_{av}(g(X))$ and raised the
following conjecture:
\begin{conj}
Is it true that for the function given by (1) we have
$$\sigma_{av}(g(X))\geq \left( \frac 12 +o(1)\right)m?$$
\end{conj}

In the same paper Shparlinski gave a lower bound by obtaining a
nontrivial bound on the Fourier coefficients of $g(X)$ via analytical methods. He proved in the same paper
that this value is greater than $\gamma m$, where $\gamma\approx
0.0575$ is a constant.
Li \cite{Li} solved this conjecture.
\begin{thm}Let $\sigma_{av}(g)$ be the average sensitivity of the
previous weighted sum function $g(X)$. Then
$$\sigma_{av}(g(X))=\left( \frac 12 +o(1)\right)m.$$
\end{thm}

 Here we prove that this
conjecture still holds for $f(X)$:
\begin{thm}Let $\sigma_{av}(f)$ be the average sensitivity of the
simplified weighted sum function $f(X)$. Then
$$\sigma_{av}(f(X))=\left( \frac 12 +o(1)\right)m.$$
\end{thm}
\begin{proof}
Since we have the symmetry between the bits 1 and 0,
 for simplicity we just need to consider the number of bit changes
from 0 to 1. Thus by Corollary \ref{cor3.4},
 \begin{align*}
2^{m-1}\sigma_{av}(f(X))&=\sum_{X\in
\mathbb{Z}_2^m}\sum_{i=0}^{m-1}\left |
 f(X)-f(X^{(i)})\right |\\
 &=\sum_{s=0}^{m-1}\sum_{X\in \mathbb{Z}_2^m, s(X)=s, x_s=1}\sum_{i=0}^{m-1}\left |
    1-f(X^{(i)})\right |\\
    &\ \ \ +\sum_{s\in D}\sum_{X\in \mathbb{Z}_2^n, s(X)=s, x_s=0}\sum_{i=0}^{m-1}\left |
 0-f(X^{(i)})\right | \\
 &=\sum_{i=0}^{m-1}\sum_{s=0}^{m-1}  \sum_{X\in \mathbb{Z}_2^m, s(X)=s, x_i=0, x_s=1,
 x_{s+i}=0}1\\
 &\ \ \ +\sum_{i=0}^{m-1}\sum_{s=0}^{m}  \sum_{X\in \mathbb{Z}_2^m, s(X)=s, x_i=0, x_s=1,
 x_{s+i}=1}1\\
 &=\sum_{i=0}^{m-1}\sum_{s=0}^{m-1}N(0,\mathbb{Z}_m\backslash \{i, s+i, s \})+
 \sum_{i=0}^{m-1}\sum_{s=0}^{m-1}N(0,\mathbb{Z}_m\backslash \{i, s-i, s \})\\
 &=\sum_{i=0}^{m-1}\sum_{s=0}^{m-1} \frac {2^{m-2}} m(1+o(1)) \\
 &= m{2^{m-2}} (1+o(1)).
 %+\sum_{s=1}^n  \sum_{X\in Z_2^n, s(X)=s, x_s=0}\sum_{i=1}^n\left |
% (x_i+x_{s+i})\right |\\
% \sum_{i=1}^n\sum_{s=n+1}^p\sum_{X\in Z_2^n, s(X)=s}\left |
% x_1-f(X^{(i)})\right |\\
% =\sum_{i=1}^n\sum_{s=1}^n\left ( \sum_{X\in Z_2^n, s(X)=s, x_i=0}\left |
% x_s-f(X^{(i)})\right |+ ( \sum_{X\in Z_2^n, s(X)=s, x_i=1}\left |
% x_s-f(X^{(i)})\right | \right )+\sum_{i=1}^n\sum_{s=n+1}^p\sum_{X\in Z_2^n, s(X)=s}\left |
% x_1-f(X^{(i)})\right |\\
\end{align*}
Finally we have $$\sigma_{av}(f(X))=\left(\frac 12
+o(1)\right)m.\qedhere$$
%We note that we only compute the cases that $f(X)$ flips from 1 to 0
%since the cases from 0 to 1 are similar.
%\begin{eqnarray*}
% &=&\sum_{i=1}^n\sum_{s\in D}\sum_{X\in Z_2^n,s(X)=s,x_s=1}\left |
% f(X)-f(X^{(i)})\right|\\
% &=& \sum_{i=1}^n\sum_{s\in D}\sum_{X\in Z_2^n,s(X)=s,x_s=1,x_i=0}(
% 1-x_{b+i \mod p})+\sum_{i=1}^n\sum_{s\in D}\sum_{X\in Z_2^n,s(X)=s,x_s=1,x_i=1}(
% 1-x_{b-i \mod p})\\
%\end{eqnarray*}
\end{proof}

%{\bf Acknowledgements.} The author wishes to thank the reviewers for
%their several constructive comments.

\section{Further questions}

\begin{table}
\center
\caption{The relationship between the variable number $m$ and the maximal Fourier coefficients of $f(X)$ over $0<m<22$. Decimals are rounded to three decimal places.}
\begin{tabular} {|c|c|c|}

\hline
\multirow{2}*{$m$}&\multirow{2}*{$\max|\widehat{f}(a)|$}&\multirow{2}*{$m^{-1}log_{2}\max|\widehat{f}(a)|$}\\
&&\\ \hline
1&      1&      0\\ \hline
2&      0.5&        -0.5\\ \hline
3&      0.5&        -0.333\\ \hline
4&      0.75&       -0.104\\ \hline
5&      0.75&       -0.083\\ \hline
6&      0.438&      -0.199\\ \hline
7&      0.469&      -0.156\\ \hline
8&      0.281&      -0.229\\ \hline
9&      0.305&      -0.191\\ \hline
10&     0.227&      -0.214\\ \hline
11&     0.146&      -0.252\\ \hline
12&     0.209&      -0.188\\ \hline
13&     0.093&      -0.264\\ \hline
14&     0.086&      -0.253\\ \hline
15&     0.159&      -0.177\\ \hline
16&     0.067&      -0.244\\ \hline
17&     0.059&      -0.240\\ \hline
18&     0.119&      -0.171\\ \hline
19&     0.053&      -0.224\\ \hline
20&     0.050&      -0.216\\ \hline
21&     0.089&      -0.166\\ \hline

\end{tabular}
\label{table1}
\end{table}

In \cite{Sh} Shparlinski studied the Fourier coefficients of the weighted sum function.
\begin{defn}Let $h(X)$ be a Boolean function from $\{ 0, 1\}^m$ to
$\{0,1\}$. The Fourier coefficient of $h(X)$ at $a$ is defined by
$$\widehat{h}(a)=\frac 1 {2^m}\sum _{X\in \{0, 1\}^m}(-1)^{h(X)+a\cdot X}.$$
\end{defn}

 Shparlinski raised the following conjecture which is stronger than his conjecture on the average sensitivity.
\begin{conj}[Shparlinski,  \cite{Sh}]
For the previous weighted sum function $g(X)$, we have
$$\max|\widehat{g}(a)|=2^{(-\frac 12 +o(1))m}.$$
\end{conj}

Shparlinski proved in \cite{Sh} that
$$\max|\widehat{g}(a)|\leq 2^{(-\rho+o(1))m},$$
where $$\rho=\frac {4}{\pi\ln{2}}\mathcal{L}(\frac{\pi}4), \ \ \mathcal{L}(x)=-\int_{0}^{x}\ln\cos{\theta}d\theta.$$

Note that $\rho\approx 0.1587$.  Currently this is still the best result.

We compute the value of $\max|\widehat{f}(a)|$ over $0<m<22$ using a computer
program. The results are shown in Table \ref{table1}.
Note again that when $m$
is prime then  $f(X)=g(X)$.  The experimental results indicate
that Shparlinski's conjecture may not be true. It will be amazing
to obtain the true bounds on both $\max|\widehat{g}(a)|$ and $\max|\widehat{f}(a)|$.

Instead of the Shparlinski's conjecture, we propose a new conjecture:
  \begin{conj}\label{conj5.3}
For the newly defined weighted sum function $f(X)$, we have
$$\max|\widehat{f}(a)|=2^{(-\rho +o(1))m}.$$
\end{conj}

\end{document}